\documentclass[12pt]{amsart}
\usepackage{amsmath,amssymb}
\usepackage{amsmath}
\usepackage{amsopn}
\usepackage{color}
\usepackage{fixltx2e}
\usepackage{url}

\numberwithin{equation}{section} %% Comment out for sequentially-numbered
\numberwithin{figure}{section} %% Comment out for sequentially-numbered

\theoremstyle{plain}
\newtheorem{thm}{Theorem}
\newtheorem{lem}[thm]{Lemma} %%Delete [thm] to re-start numbering
\theoremstyle{definition}

\newtheorem{asst}{Assumption}

\theoremstyle{remark}
 \newtheorem*{rem*}{Remark}

\newcommand{\ip}[1]{\langle #1 \rangle}
\newcommand{\Ev}[1]{\E \left( #1 \right)}  %% produces \E( # )
\newcommand{\norm}[1]{\left\Vert#1\right\Vert}
\newcommand{\abs}[1]{\left\vert#1\right\vert}
\newcommand{\set}[1]{\left\{#1\right\}}

\renewcommand{\vec}[1]{\mathbf{#1}}

\newcommand{\bb}[1]{\mathbb{#1}}
\newcommand{\cu}[1]{\mathcal{#1}}

\newcommand{\wh}[1]{\widehat{#1}}

\def\Var{\operatorname{Var}}

\def\e{\mathrm e}

\def\im{\mathrm i}
\def\Im{\mathrm{Im}}
\def\half {\frac{1}{2}}
\def\1{{\mathsf 1}}
\def\di{\mathrm d}
\def\grad{\nabla}
\def\wt{\widetilde}

\makeatletter
\def\rightharpoondownfill@{%
    \arrowfill@\relbar\relbar\rightharpoondown}
\def\rightharpoonupfill@{%
    \arrowfill@\relbar\relbar\rightharpoonup}
\def\leftharpoondownfill@{%
    \arrowfill@\leftharpoondown\relbar\relbar}
\def\leftharpoonupfill@{%
    \arrowfill@\leftharpoonup\relbar\relbar}
\newcommand{\xrightharpoondown}[2][]{%
    \ext@arrow 0359\rightharpoondownfill@{#1}{#2}}
\newcommand{\xrightharpoonup}[2][]{%
    \ext@arrow 0359\rightharpoonupfill@{#1}{#2}}
\newcommand{\xleftharpoondown}[2][]{%
    \ext@arrow 3095\leftharpoondownfill@{#1}{#2}}
\newcommand{\xleftharpoonup}[2][]{%
    \ext@arrow 3095\leftharpoonupfill@{#1}{#2}}
\newcommand{\xleftrightharpoons}[2][]{\mathrel{%
    \raise.22ex\hbox{%
        $\ext@arrow 3095\leftharpoonupfill@{\phantom{#1}}{#2}$}%
    \setbox0=\hbox{%
        $\ext@arrow 0359\rightharpoondownfill@{#1}{\phantom{#2}}$}%
    \kern-\wd0 \lower.22ex\box0}%
}
\newcommand{\xrightleftharpoons}[2][]{\mathrel{%
    \raise.22ex\hbox{%
        $\ext@arrow 3095\rightharpoonupfill@{\phantom{#1}}{#2}$}%
    \setbox0=\hbox{%
        $\ext@arrow 0359\leftharpoondownfill@{#1}{\phantom{#2}}$}%
    \kern-\wd0 \lower.22ex\box0}%
} \makeatother

\def\Z{\mathbb Z}
\def\N{\mathbb N}
\def\R{\mathbb R}
\def\C{\mathbb C}
\def\E{\mathbb E}

\def\ra{\rightarrow}

\def\ran{\operatorname{ran}}

\def\tr{\operatorname{tr}}    %%% Trace
\def\Re{\operatorname{Re}}
\def\Im{\operatorname{Im}}
\def\tem{\textemdash \ }

\title{Diffusive propagation of wave packets in a fluctuating periodic potential}
\author {Eman Hamza}
\address{Department of Mathematics \\ Michigan State University \\ East Lansing, MI 48823}
\curraddr{Department of Physics\\ Cairo University \\ Cairo, Egypt}

\author{Yang Kang}
\address{Department of Mathematics \\ Michigan State University \\ East Lansing, MI 48823}

\author{Jeffrey Schenker}
\address{Department of Mathematics \\ Michigan State University \\ East Lansing, MI 48823}
\email{jeffrey@math.msu.edu}
\thanks{The third author was supported by an NSF CAREER Award DMS-08446325.}

\date{2 February, 2010, revised 26 June 2010, minor corrections 4 October 2010}
\begin{document}
\begin{abstract}  We consider the evolution of a tight binding wave packet propagating in a fluctuating
periodic potential. If the fluctuations stem from a stationary Markov process satisfying certain technical criteria, we 
show that the square amplitude of the wave packet after diffusive rescaling converges to a  superposition of solutions of a heat equation. 
\end{abstract}
\maketitle
\pagestyle{plain}
\bibliographystyle{amsplain}

\section{Introduction}   

It is generally expected that wave packets evolving in a homogeneous random environment propagate diffusively over long time scales, unless recurrence effects are strong enough to induce Anderson localization.  If furthermore the environment fluctuates in time, recurrence effects should be irrelevant, suggesting that diffusion is universal for wave motion in time dependent random systems.  This idea was confirmed by Ovchinnikov and Erikman \cite {Ovchinnikov:1974eu}, who showed diffusion for a tight binding Schr\"odinger equation with white noise potentials. Pillet \cite{Pillet:1985oq} considered a more general setting in which the potentials were Markov processes, but not necessarily white noise.  He demonstrated the absence of binding and derived a Feynman-Kac formula. This Feynman-Kac formula was used by Tcheremchantsev \cite{Tcheremchantsev:1997kl, Tcheremchantsev:1998qe} to show that position moments scale diffusively up to logarithmic corrections. Recently, two of us  proved diffusion of wave packets and diffusive scaling  \cite{KS} for a large family of Markov models, including those considered by Tcheremchantsev.  For a recent discussion of the physics and physical applications of the tight binding Schr\"odinger equation with time dependent randomness  we refer to \cite{Amir:2009}.

The study of diffusion for disordered quantum systems, or ``Quantum Brownian motion," has recently attracted the attention of a number of authors.  Diffusion  in the presence of  a weak static random potential  for a quantum particle on a lattice of dimension three or higher has been demonstrated only up to a finite time scale proportional to an inverse power of the disorder strength \cite{Erdos2007:621, Erdos2007:1, Erdos2008:211}. Fr\"ohlich, Pizzo and De Roeck have proved diffusion, for arbitrarily long times, for a  quantum particle on a lattice weakly coupled to an array of independent heat baths \cite{W.-De-Roeck:rz}.  In \cite{W.-De-Roeck:rz}, it is mentioned that the method used therein also applies to a particle in a time dependent potential provided one has exponential decay of time correlations, such as one has for the Markov potentials in \cite{KS}.  However the proof in \cite{W.-De-Roeck:rz} relies on a polymer expansion which restricts the result to weak coupling (this is not the case in \cite{KS}).   Another result closely related to our previous work \cite{KS} is a  recent paper on diffusion starting from a quantum master equation in Lindblad form \cite{Clark:2008}.

This note and the aforementioned \cite{Pillet:1985oq,Tcheremchantsev:1997kl, Tcheremchantsev:1998qe,KS} are concerned with the evolution of wave packets for the ``tight binding Markov random Schr\"odinger equation:''
\begin{align}\label{Schrodinger}
\begin{cases} \im \partial_{t} \psi_{t}(x) \ = \ L \psi_{t}(x) + v_{x}(\omega(t)) \psi_{t}(x), \\
\psi_{0} \in \ell^{2}(\Z^{d}),
\end{cases}
\end{align} 
where
\begin{enumerate}
\item  $L$ is a translation invariant hopping operator on $\ell^{2}(\Z^{d})$, 
\item $v_{x}:\Omega \rightarrow \R$ are real valued functions on a probability space $\Omega$,
\item $\omega(t)$ is a Markov process on $\Omega$ with an invariant probability measure $\mu$, and
\item $v_{x}(\omega)=v_{0}(\sigma_{x}(\omega))$ where $\sigma_{x}$ is a group of $\mu$-measure preserving transformations of $\Omega$.
\end{enumerate}
(Formal definitions are given in section \ref{sec:main} below.)

The potentials considered by Tcheremchantsev \cite{Tcheremchantsev:1997kl, Tcheremchantsev:1998qe} were independent at different sites.  However, this played no role in the analysis in \cite{KS}. Nonetheless, some non-degeneracy assumption is certainly needed as can be seen by considering the case $v_{x}=v_{0}$ for all $x$, for which  the effect of the random potential is only to multiply the wave function by a time dependent random phase. The technical condition employed in \cite{KS} was
\begin{align}
\inf_{x} \norm{B^{-1}(v_x  - v_0)}_{L^{2}(\Omega)}  >0 \label{v-non-degenerate},
\end{align}  
where $B$ is the generator of the Markov process $\omega(t)$.  

Our aim here is to consider a situation in which \eqref{v-non-degenerate} is violated in a relatively strong way.  Namely, we shall consider \emph{periodic} potentials, $v_{x+Ny} = v_{x}$ for all $x,y$ with $N$ some fixed number.  Because the resulting system is periodic under translations by elements of $N \Z^{d}$, there is a conserved ``quasi-momentum.'' Our main result, in short, is that after taking into account conservation of quasi-momentum the motion of the wave packet is diffusive.  More specifically, over long times one sees a superposition of diffusions:
\begin{equation}\label{eq:mainintro}
\lim_{\tau \ra \infty} \sum_{x \in \Z^{d}} e^{-\im \frac{1}{\sqrt{\tau}} \vec{k}\cdot x } \Ev{\abs{\psi_{\tau t}(x)}^{2}} \ = \ \int_{\mathbb{T}_{N}^{d}} e^{-t \sum_{i,j=1}^{d}D_{i,j}(\vec{p}) 
\vec {k}_{i}\vec {k}_{j}} m(\vec{p}) \di \vec{p},  \end{equation}
where $\mathbb{T}_{N}^{d}=[0,2\pi/N)^{d}$, $\vec{p} \mapsto D_{i,j}(\vec{p})$ is a continuous function taking values in the positive definite matrices, independent of $\psi_{0}$, and
\begin{equation}\label{eq:mintro}
m(\vec{p}) \ = \ \frac{1}{(2 \pi )^{d}} \sum_{\zeta \in \Lambda} \abs{\wh{\psi}_{0}\left (\vec{p} + \frac{2\pi}{N} \zeta \right )}^{2}
\end{equation}
with $\Lambda =[0,N)^{d} \cap \Z^{d}$.   The quantity $m(\vec{p})$ is the amplitude of the initial wave packet at quasi-momentum $\vec{p}$ \tem \  $\wh{\psi}_{0}$ denotes the Fourier transform of $\psi_{0}$:
\begin{equation}
\wh{\psi}_{0}(\vec{k}) \ =\ \sum_{x}\e^{\im x \cdot \vec{k}} \psi_{0}(x),
\end{equation}
if $\psi_{0} \in \ell^{1} \cap \ell^{2}$.

To understand the meaning of \eqref{eq:mainintro}, consider the following position space density
\begin{equation}
\di R_{t}(x) = \sum_{\xi \in \Z^{d}} \Ev{|\psi_{t}(\xi)|^{2}}  \delta(x-\xi)\di x,
\end{equation}
a probability measure on $\R^{d}$.  (Here $\delta(x)\di x$ is the Dirac measure with mass $1$ at $0$.) After taking inverse Fourier transforms of both sides, \eqref{eq:mainintro} shows
\begin{multline}
\int_{\R^{d}} \phi(x) \di R_{\tau t}(\sqrt{\tau} x) \ \xrightarrow[\tau \ra \infty]{} \\  \int_{\R^{d}} \phi(x) \left [\int_{\mathbb{T}_{N}^{d}} 
\frac{1}{(4 \pi t)^{\frac{d}{2}} \sqrt{\det D_{i,j}(\vec{p})}}  \e^{-\frac{1}{4 t} \sum_{i,j} D_{i,j}^{-1}(\vec{p}) x_{i}x_{j}} m(\vec{p}) \di \vec{p}  \right  ]\di x,
\end{multline}
for any test function $\phi$ on $\R^{d}$ which is, say, smooth and compactly supported.  The function appearing as the integrand inside square brackets on the right hand side is the fundamental solution to an anisotripic diffusion equation, with diffusion matrix $D_{i,j}(\vec{p})$,
\begin{equation}\label{eq:diffusionequation}
\frac{\partial}{\partial t} u_{t }(x) \ = \ \sum_{i,j} D_{i,j}(\vec{p})\frac{\partial}{\partial x_{i}} \frac{\partial}{\partial x_{j}} u_{t}(x).
\end{equation}
Thus \eqref{eq:mainintro} can be understood as saying that the position space density $dR_{t} (x)$, after diffusive rescaling $t \mapsto \tau t$ and $x \mapsto \sqrt{\tau}x$, converges in the weak$^{*}$ sense to
\begin{equation}
 \di R_{\tau t}(\sqrt{\tau} x)  \xrightarrow[\tau \ra \infty]{\text{weak}^{*}}  \left [ \int_{\mathbb{T}_{N}^{d}}u_{t}(x;\vec{p})  m(\vec{p}) \di \vec{p} \right ] \di x ,
\end{equation}
where $u_{t}(x;\vec{p})$ satisfies \eqref{eq:diffusionequation} with $u_{0}(x;\vec{p}) \di x = \delta(x) \di x$.
That is over long time scales, after diffusive rescaling, the mean square amplitude breaks into components for each $\vec{p}$, with each component propagating independently and according to a diffusion equation, which is to say a ``super-position of diffusions.''

The result is stated formally in section \ref{sec:main} after we give the required assumptions.  These assumptions are somewhat technical, so it may be useful to have a simple example in mind.  Fix a function $U:\Z^{d} \rightarrow \R$ periodic under translations in $N \Z^{d}$, that is, $U(x-Ny)=U(x)$ for all $x,y \in \Z^{d}$.  Now let  $\omega(t)$ be a continuous time random walk on $\Lambda =[0,N)^{d}\cap \Z^{d}$ taken with periodic boundary conditions and with independent identically distributed exponential holding times at each step.  The probability space is just $\Lambda$ with the measure $\mu$ normalized counting measure.  Take the potentials $v_{x}$ to be $v_{x}(\omega) = U(x-\omega)$ so that the Schr\"odinger equation describes a particle in a ``jiggling'' periodic potential:
\begin{equation}
\im \partial_{t} \psi_{t}(x) \ = \ \sum_{\zeta}h(\zeta) \psi_{t}(x-\zeta) + U(x-\omega(t)) \psi_{t}(x).  
\end{equation}
Our result shows that \eqref{eq:mainintro} holds provided $U$ has no smaller periods, i.e. that $$\sum_{y \in \Lambda}|U(x+y)-U(y)| \neq 0 , \quad x \in \Lambda \text{ and }x \neq 0.$$

\section{Statement of the main result: A superposition of  diffusions  }\label{sec:main}
\subsection{Assumptions} Our main result is formulated with the following assumptions.  (See \cite{KS} for a more detailed discussion of the framework.)  

\begin{asst}
We are given a topological space $\Omega$, a Borel probability measure $\mu$, and a Markov process on $\Omega$ with right continuous paths for which $\mu$ is an invariant measure.  Furthermore, we suppose that there is a representation of $\Z^{d}$, $x \mapsto \sigma_{x}$, in terms of $\mu$-measure preserving maps $\sigma_{x}:\Omega \ra \Omega$ such that the paths of $\sigma_{x}(\omega(\cdot))$ have the same distribution as the paths of $\omega(\cdot)$, for all $x \in \Z^{d}$.
\end{asst}

We denote by $\Ev{\cdot}$ expectation with respect to the paths of the Markov process with the initial condition $\omega(0)$ distributed according to $\mu$. By the invariance of $\mu$, we have
\begin{equation}\label{eq:invariantexp}
\Ev{f(\omega(t))} = \int_{\Omega}f(\alpha) \di \mu(\alpha)
\end{equation}
for any $t$ and any $f \in L^{1}(\Omega)$.   Furthermore,  the map $S_{t}$  given by
\begin{equation} 
  S_{t}f(\alpha) \ = \ \E( f(\omega(0)) | \omega(t) =\alpha) 
\end{equation}
defines a   strongly continuous contraction semi-group on $L^{2}(\Omega)$. By the Lumer-Phillips theorem, $S_{t}$ is generated by a maximally dissipative operator $B$ with dense domain $\cu{D}(B)$.  Since $S_{t}1 =1$ for all $t$, $B1=0$ and $0$ is an eigenvalue of $B$. Since $B$ is dissipative, we also have that its numerical range lies in the right half plane.  We suppose further  that $B$ is sectorial and satisfies a ``spectral gap'' condition:
 
 \begin{asst}  There exist $\gamma < \infty$ and $T >0$ such that
\begin{align}
	  \abs{\Im \ip{f, B f}_{L^{2}(\Omega)}} \ & \le \ \gamma \Re \ip{f, B f}_{L^{2}(\Omega)} , \quad  \label{eq:sectoriality} \intertext{and}
	\Re\ip{f, B f}_{L^{2}(\Omega)} & \ge \frac{1}{T} \Var ( f) \label{eq:spectralgap}
\end{align}
for all $f \in \cu{D}(B)$, where $\Var(f):=\int_{\Omega} f^{2}\di \mu - (\int_{\Omega} f \di \mu)^{2}$.\end{asst}

The potential $v_{x}:\Omega \ra \R$ and hopping operator $L$ are assumed to be translation invariant, and $L$ should satisfy a non-degeneracy condition that precludes hopping only in a sub-lattice:
\begin{asst}
The potential is given by Borel measurable bounded functions $v_{x} :\Omega \ra \R$  such that
$$ \quad v_{x} = v_{0} \circ \sigma_{x} .  $$
The hopping operator is given by
$$L \psi(x) \ = \ \sum_{y} h(x-y) \psi(y),$$
where $h(-x)=h(x)^{*}$, $ \sum_{x} |x|^{2} \abs{h(x)} <  \infty$,  and for each  non-zero vector $\vec k \in \R^{d}$, there is some $x \in \Z^{d}$ such that $h(x) \neq 0$ and $\vec k \cdot x  \ne 0$.
\end{asst}

Finally, since we are concerned with periodic potentials, we suppose
\begin{asst} There is $N \in \N$, $N > 1$, such that $\sigma_{Nx}= \operatorname{Id}$ for all $x \in \Z^{d}$.  Furthermore, we suppose that  $\norm{v_{x}-v_{0}}_{L^{\infty}(\Omega)}>0$ for all  $x \in [0,N)^{d}\cap \Z^{d}$, $x \neq 0$.
\end{asst}
\begin{rem*} More generally, we might allow different periods in each of the coordinate directions: $N_{1}, \ldots N_{d}$ such that $\sigma_{y} = \operatorname{Id}$ whenever $y = (N_{1} \alpha_{1}, \ldots, N_{d}\alpha_{d})$ with $\alpha_{1}, \ldots, \alpha_{d} \in \Z$.  The result stated below holds also for this case with essentially the same proof.  We choose to work with equal periods for notational clarity.
\end{rem*}

Let $\Lambda  = [0,N)^d \cap \Z^{d}$, as above, and let $x,y\in\Lambda$.   Since $v_{x}-v_{y}$ is mean zero, it is in the domain of $B^{-1}$.  Furthermore, it follows from Assumption 4 that $\norm{v_{x}-v_{y} }_{L^{2}(\Omega)}\neq 0$ if $x\neq y$, in which case $\norm{B^{-1}(v_{x}-v_{y})}_{L^{2}(\Omega)}\neq 0$. Since $\Lambda$ is finite, we conclude that there is $\chi > 0$ such that
\begin{equation}\label{eq:nondeg}
\norm{B^{-1} (v_{x}- v_{y})}_{L^{2}(\Omega)} \ge \chi \ , \quad x,y \in \Lambda , \ x \neq y.
\end{equation}
Eq. \eqref{eq:nondeg} will play a key role in the proof below.

\subsection{Main result}
Consider the density matrix
\begin{equation}
\rho_{t}(x,y) \ = \ \psi_{t}(x) \psi_{t}(y)^{*}.
\end{equation}
It is well-known that  $\rho_t(x,y)$ satisfies
\begin{equation}\label{MAMDM}
\partial_{t} \rho_{t}(x,y) \ = \ - \im \sum_{\zeta} h(\zeta) \left[ \rho_{t}(x-\zeta,y) - \rho_{t}(x,y+\zeta) \right ]
-\im \left ( v_{x}(\omega(t)) - v_{y}(\omega(t)) \right ) \rho_{t}(x,y).
\end{equation}
More generally, we may consider solutions to \eqref{MAMDM} with an initial condition 
\begin{multline}
\rho_{0} \in \cu{DM} \ := \ \left \{ \rho: \Z^{d} \times \Z^{d} \ra \C \ : \ \rho \text{ is the kernel of a non-negative definite,} \right . \\
\left.  \text{trace class operator on $\ell^{2}(\Z^{d})$} \right \}.
\end{multline}
Recalling the notation $\mathbb{T}_{N}^{d}=[0,2\pi/N)^{d}$, we now state our theorem.
 \begin{thm}\label{thm:main} 
The solution to \eqref{MAMDM} with initial condition $\rho_{0} \in \cu{DM}$ satisfies
\begin{equation}\label{eq:main}
\lim_{\tau \ra \infty }\sum_{x} \e^{-\im \frac{\vec{k}}{\sqrt{\tau}} \cdot x} \Ev{\rho_{\tau t} (x,x)} \ = \   \int_{\bb{T}_{N}^d}  \, \e^{-t  \sum_{i,j}D_{i,j}(\vec p)\vec{k}_{i} \vec{k}_{j}}  m(\vec p)\di \vec{p},
\end{equation}
where $\vec{p} \mapsto D_{i,j}(\vec p)$  is a continuous function taking values in the positive-definite matrices and 
$$m(\vec{p}) \ = \ \frac{N^{d}}{(2\pi)^{d}} \wh{f}(N\vec{p}) , $$ with $\wh{f}$ the Fourier transform of 
$f(x) = \sum_{y \in \Z^{d}} \rho_{0}(y+Nx,y) .$ 
\end{thm}

\begin{rem*}We have defined the function $m$ in terms of the Fourier transform of $f$.  Since $f$ is not obviously summable or square summable, it is not immediately clear that $m$ is indeed a \emph{function}, rather than a distribution. However, in terms of the orthnormal eigenvectors $\psi_{j}$ of $\rho_{0}$ and corresponding eigenvalues $\lambda_{j}$, we have
\begin{equation}\label{eq:nobochner}
m(\vec{p}) \ = \   \frac{1}{(2\pi)^{d}} \sum_{j} \lambda_{j} \sum_{\zeta \in \Lambda} \abs{\wh{\psi_{j}}\left (\vec{p} + \frac{2 \pi}{N} \zeta \right )}^{2}.
\end{equation}
Since $\sum_{j} \lambda_{j} <\infty$ and $\abs{\wh{\psi_{j}}}^{2} \in L^{1}(\mathbb{T}^{d})$ we see that $m(\vec{p})$ is an $L^{1}$ function of $\vec{p}$.  (The function $f(x)$ can be expressed as
\begin{equation}
f(x) \ = \ \tr \rho_{0} S_{Nx}
\end{equation}
where $\rho_{0}$ is interpreted as a trace class operator and $S_{Nx}$ is the shift by $Nx$ on $\ell^{2}(\Z^{d})$, $S_{Nx}\psi(y)=\psi(y-Nx)$. It follows that $f$ is \emph{positive definite}:
\begin{equation}
\sum_{i,j=1}^{n}\zeta_{i}^{*}\zeta_{j}f(x_{i}-x_{j})\ge 0
\end{equation}
for any finite collection of points $x_{1},\ldots, x_{n} \in \Z^{d}$ and any $(\zeta_{1},\ldots, \zeta_{n})\in \C^{n}$. We conclude from Bochner's theorem that $\wh{f}$ is a non-negative measure of mass $f(0)=\tr\rho_{0}$, and  because $\lim_{x\rightarrow \infty}f(x)=0$ the measure has no point component.  But, it is not immediately clear that $\wh{f}$ is absolutely continuous with respect to Lebesuge measure so that $m$ is a function.  For this purpose \eqref{eq:nobochner} seems to be necessary.)
\end{rem*}

\section{Augmented space analysis}
In this section, we explain briefly the augmented space analysis, which was also employed in \cite[Section 3]{KS}. We begin with the following Feynman-Kac formula  \cite{Pillet:1985oq} 
\begin{equation}
	\E(\rho_t (x,y)) = \ip{ \delta_{x}\otimes \delta_{y} \otimes 1,  e^{-t L  } \rho_0 \otimes 1}_{\cu{H}}, \label {E(Rho)}
\end{equation}
which  relates $\Ev{\rho_{t}(x,x)}$ to a matrix element of a contraction semigroup $e^{-tL}$ on the augmented Hilbert space 
 \begin{equation}
 \cu{H} \ := \ L^{2}(\Z^{d} \times \Z^{d} \times \Omega). 
\end{equation}
The operator    $L$  in \eqref{E(Rho)} is given by
  $ L := \im K + \im V +B$,
where
\begin{align}
K \Psi(x,y,\omega) \ &= \    \sum_{\zeta} h(\zeta) \left[ \Psi(x-\zeta,y,\omega) - \Psi(x,y+\zeta,\omega) \right ], \\
V \Psi(x,y,\omega) \ &= \  \left ( v_{x}(\omega) -v_{y}(\omega) \right ) \Psi(x,y,\omega).
\end{align}
The Markov generator $B$ acts on $\cu{H}$ as a multiplication operator with respect to the first two coordinates: 
\begin{equation}
B [\rho \otimes f] \ = \ \rho \otimes (Bf), \quad \rho \in \ell^{2}(\Z^{d} \times \Z^{d}) , \ f \in L^{2}(\Omega).
\end{equation}

Our analysis, as in \cite{KS}, makes crucial use of the invariance of the generator $L$ with respect to simultaneous translation of position and disorder.   In the present context, we have a larger group of symmetries due to periodicity.  Namely, the generator $L$  and its constituents $K$, $V$, and $B$, commute with a group $\cu{G}$ of unitary maps on $\cu{H}$ generated by the following transformations:
\begin{enumerate}
\item Simultaneous translation of position and disorder by an arbitrary element of $\Z^{d}$:
$$ S_{\xi}\Psi(x,y,\omega)=\Psi(x-\xi,y-\xi,\sigma_{\xi}\omega),$$
\item Translation of the first position coordinate by an element of $N \Z^{d}$:
$$ S^{(1)}_{N \xi}\Psi(x,y,\omega)=\Psi(x-N\xi, y, \omega).$$
\end{enumerate}
Note that $S_{\xi}S^{(1)}_{N\eta}=S^{(1)}_{N\eta}S_{\xi}$, so the group $\cu{G}$ is isomorphic to $\Z^{d}\times \Z^{d}$.  We have chosen to use translation of the first position in the definition of $S^{(1)}$; however, since $\sigma_{N\xi}=\operatorname{Id}$, we have $S^{(2)}_{N\xi}=S_{N\xi}S^{(1)}_{-N\xi} \in \cu{G}$, where
$ S^{(2)}_{N \xi}\Psi(x,y,\omega)=\Psi(x,y-N\xi, \omega).$

Because of the invariance with respect to $\cu{G}$, $L$ is partially diagonalized by the following generalized Fourier transform:
\begin{equation}
\wt{ \Psi}(x,\omega, \vec{k}, \vec{p})  = \sum_{\xi,\eta \in \Z^{d} } \e^{\im \vec{p} \cdot (x-N\eta) -\im \vec{k}\cdot \xi}  \Psi (x -\xi -N\eta , -\xi  , \sigma_{\xi}\omega), \end{equation}
a unitary map from  $L^2(\Z^{d} \times \Z^{d} \times \Omega) \to  L^2(\Lambda \times \Omega \times \bb{T}_{1}^{d} \times \bb{T}_{N}^d).$ 
Thus we have, by \eqref{E(Rho)},
\begin{equation}
\sum_{x} \e^{-\im \vec{k} \cdot x} \Ev{\rho_{t}(x,x)} \ = \ \frac{N^{d}}{(2\pi)^{d}}\int_{ \bb T_{N}^d} \di \vec{p}  \ip{ \delta_{0} \otimes 1, \e^{-t \wt{L}_{\vec{k} , \vec{p}}} \wt{\rho}_{0;\vec{k}, \vec{p}}\otimes 1}_{L^{2}(\Lambda \times \Omega)} \label {E(Rho3)FT},
\end{equation}
where
 \begin{align}
\wt{\rho}_{0;\vec{k},\vec{p}}(x) \ &= \ \sum_{y, \eta} \e^{\im \vec{p}\cdot(x-N \eta)-\im \vec{k}\cdot y} \rho_{0}(x-N\eta -y,-y),  
\end{align}
and $\wt{L}_{\vec{k} , \vec{p}}  \ := \ \im \wt{K}_{\vec{k} , \vec{p}} + \im \wt{V} + B$
with
\begin{align}
    \wt{V} \wt{\psi}(x,\omega) &=    (v_{x}(\omega) -v_{0}(\omega))\wt \psi (x, \omega),
    \intertext{and}
    \wt{K}_{\vec{k}, \vec p} \wt{\psi}(x,\omega )  &=  \sum_{\zeta} h(\zeta)\e^{\im \vec {p} \cdot \zeta} \left [ 
     \wt{\psi}(x-\zeta,\omega) - \e^{-\im \vec{k} \cdot \zeta}
     \wt{\psi}(x-\zeta, \sigma_{\zeta}\omega) \right]. \label{eq:wtK}
\end{align}
(In \eqref{eq:wtK} we take ``periodic boundary conditions,''  that is $x -\zeta$ on the right hand side is evaluated modulo $N$.)

The transformed Feynmann-Kac formula \eqref{E(Rho3)FT} is the starting point for our proof of Theorem \ref{thm:main}.  It reduces the study of the  mean density in   \eqref{eq:main} to the spectral analysis of the semi-group $\e^{-t \wt{L}_{\vec{k}, \vec{p}}}$  for each fixed $\vec p$ and for $\vec{k}$ in a small neighborhood of $0$.

\section{Spectral analysis of $\wt{L}_{\vec{k} , \vec{p}}$ and the proof of Theorem \ref{thm:main}}
In this section inner products and norms are taken in the space $L^{2}(\Lambda \times \Omega)$ unless otherwise indicated. We denote by $P_{0}$ the  orthogonal projection of $L^{2}(\Lambda \times \Omega)$ onto the space $\cu{H}_{0}= \ell^{2}(\Lambda)\otimes \{1\}$ of ``non-random'' functions,
\begin{equation}
P_{0}\Psi(x) \ = \ \int_{\Omega } \Psi(x,\omega) \di \mu(\omega), 
\end{equation} 
and  by $P_{0}^{\perp}=(1-P_{0})$  the projection onto mean zero functions 
\begin{equation}
\cu{H}^{\perp}_{0} = \set{ \Psi(x,\omega) \ : \ \int_{\Omega} \Psi(x,\omega) \di\mu(\omega) = 0}.
\end{equation} 

A preliminary observation is that
\begin{equation}\label{eq:groundstate}
\wt{L}_{\vec{0},\vec{p}} \delta_{0} \otimes 1 = 0 
\end{equation}
for all $\vec{p}$. Thus,  $\delta_{0} \otimes 1$ is stationary under each semigroup $e^{-t \wt{L}_{\vec{0},\vec{p}}}$.  Eq.\ \eqref{eq:groundstate} can be seen easily from the explicit form for $\wt{L}_{\vec{0},\vec{p}}$ given above, but could also be derived from the fact that, for each $y \in \Z^{d}$,
$$\sum_{x} \Ev{\rho_{t}(x+Ny,x)}$$
is constant in time.

A key step toward proving Theorem \ref{thm:main} is to observe that the remaining spectrum of $\wt{L}_{\vec{0},\vec{p}}$ is contained in a half plane with strictly positive real part.  To see this, we make use of the block decomposition of $\wt{L}_{ \vec{0}, \vec{p}}$   with respect to the
direct sum $\cu{H}_{0} \oplus \cu{H}_{0}^{\perp}$:
\begin{equation}\label{eq:blockform}
\wt{L}_{\vec{0} , \vec{p}}  \ = \ \begin{pmatrix}
0 & \im P_{0} \wt{V} \\
\im  \wt{V} P_{0} & \im \wt{K}_{\vec{0}, \vec{p}} + B + \im P_{0}^{\perp} \wt{V} P_{0}^{\perp}
\end{pmatrix}.
\end{equation}
(Note that $\wt{K}_{\vec{0},\vec{p}}$ and $B$ both act trivially on $\cu{H}_{0}$, while $P_{0} \wt{V} P_{0}=0$ since $\int_{\Omega}(v_{x}(\omega) -v_{0}(\omega) ) d \mu(\omega)=0$.)
 
We use  \eqref{eq:blockform} to prove the following 
\begin{lem}\label{lem:L0}
There is $\delta >0$ such that for all $\vec{p} \in \bb{T}^{d}_{N}$,
\begin{equation}
\sigma(\wt{L}_{\vec{0}, \vec{p}}) \ = \ \{0\} \cup \Sigma_{+}(\vec{p}) 
\end{equation}
where $0$ is a non-degenerate eigenvalue and  $\Sigma_{+}(\vec{p}) \subset \set{z \ : \ \Re z >\delta}.$ 
%In addition, $\e^{-t \wt{L}_{\vec{0} , \vec{p}}} (1-Q_{0})$
%is a contraction semi-group on $\ran (1-Q_{0})$, and for all sufficiently small $\epsilon >0$  there is $C_{\epsilon} > 0$ such that
%\begin{equation}\label{eq:expdecay}
%\norm{\e^{-t \wt{L}_{\vec{0} , \vec{p}}} (1- Q_{0})} \ \le \ C_{\epsilon} \e^{-t (\delta_{\lambda} -\epsilon)}.
%\end{equation}
%\begin{enumerate}
%\item $0$ is a non-degenerate eigenvalue, and
%\item $\Sigma_{+} \subset \set{z \ : \ \Re z >\delta_{\lambda}}.$
%\end{enumerate}
\end{lem}

\begin{proof} This is very close to \cite[Lemma 3]{KS}. The key new point is that we must see that $\delta$ can be chosen independently of $\vec{p}$.

Because $\Re B \ge \frac{1}{T} P_{0}^{\perp}$, it follows from an argument using Schur complements that a point $z$ with $\Re z < \frac{1}{T}$ is in $\sigma(\wt{L}_{\vec{0},\vec{p}})$ if and only if $z$ is in the spectrum of
\begin{equation}
\Gamma_{\vec{p}}(z) = P_{0} \wt{V} (P_{0}^{\perp} \wt{L}_{\vec{0},\vec{p}} P_{0}^{\perp}-z)^{-1} \wt{V} P_{0}. 
\end{equation}
However, given $\phi \in \ell^{2}(\Lambda)$,
\begin{align}
\Re & \ip{\phi\otimes 1, \Gamma_{\vec{p}}(z) \phi \otimes 1} \notag \\
&= \left \langle (P_{0}^{\perp} \wt{L}_{\vec{0} ,  \vec{p}} P_{0}^{\perp}-z)^{-1} \wt{V} \phi \otimes 1 \, ,\,  (\Re B - \Re z) (P_{0}^{\perp} \wt{L}_{\vec{0},\vec{p}} P_{0}^{\perp}-z)^{-1} \wt{V} \phi \otimes 1 \right \rangle  \notag \\
& \ge \left ( \frac{1}{T} - \Re z \right ) \norm{ (B^{-1} P_{0}^{\perp}  (\wt{L}_{\vec{0},\vec{p}} -z )P_{0}^{\perp})^{-1} B^{-1} \wt{V} \phi \otimes 1}^{2},
\end{align}
where the inverses are well defined because $\wt{V} \phi \otimes 1 \in \cu{H}_{0}^{\perp}=\ran P_{0}^{\perp}$. Since $\norm{B^{-1} P_{0}^{\perp}} \le T$, it follows that
\begin{equation}
\norm{B^{-1} P_{0}^{\perp}  (\wt{L}_{\vec{0},\vec{p}} -z )P_{0}^{\perp} } \ \le \ 1 + T\left (\| \wt{K}_{\vec{0},\vec{p}} \| + \| \wt{V} \|+ |z| \right ).
\end{equation}
However, $\|\wt{K}_{\vec{k},\vec{p}}\|\le 2 \|\wh{h}\|_{\infty}$ for all $\vec{k}$ and $\vec{p}$, so $B^{-1} P_{0}^{\perp}  (\wt{L}_{\vec{0},\vec{p}} -z )P_{0}^{\perp}$ is uniformly bounded and 
\begin{align}
\Re & \ip{\phi\otimes 1, \Gamma_{\vec{p}}(z) \phi \otimes 1} \notag \\
& \ge \left ( \frac{1}{T} - \Re z \right ) \frac{1}{ \left [1 + T(2 \| \wh{h}\|_{\infty} + 2\|\wt{V}\| + |z|) \right ]^{2}} \norm{  B^{-1} \wt{V} \phi \otimes 1}^{2}.
\end{align}
Finally,
\begin{equation}
\norm{  B^{-1} \wt{V} \phi \otimes 1}^{2} \ = \ \sum_{x}|\phi(x)|^{2} \norm{B^{-1}(v_{x}-v_{0})}_{L^{2}(\Omega)}^{2}  \ \ge \ \chi^{2}
\sum_{x \neq 0}|\phi(x)|^{2}, 
\end{equation}
where
\begin{equation}
\chi \ = \ \min_{\substack{x \in \Lambda \\ x \neq 0}} \norm{B^{-1}(v_{x}-v_{0})}_{L^{2}(\Omega)},
\end{equation}
which is positive by Assumption 4.

Thus,
\begin{equation}
\Re \Gamma_{\vec{p}}(z) \ \ge \ \left ( \frac{1}{T} - \Re z \right ) \frac{\chi^{2}}{\left [1 + T(2 \| \wh{h}\|_{\infty} + 2\|\wt{V}\| + |z|) \right ]^{2}}.
\end{equation}
Since the right hand side is independent of $\vec{p}$, the existence of a spectral gap $\delta$ independent of $\vec{p}$, as claimed, now follows from the sectoriality of $B$ (Assumption 2, eq.\ \eqref{eq:sectoriality}) as in the proof of \cite[Lemma 3]{KS}, with the explicit estimate
\begin{equation}
\delta \ge \frac{1}{T} \frac{\chi^{2 }}{\left ( 2 + \gamma + 4 T \|\wh{h}\|_{\infty} + 4 T \|\wt{V}\|\right )^{2}+ \| \wt{V}\|^{2}\chi^{2}}. \qedhere
\end{equation}
\end{proof}

\subsection{Analytic perturbation theory for $\wt{L}_{\vec{k} , \vec{p}}$} We now hold $\vec{p}$ fixed and consider the spectrum of $\wt{L}_{\vec{k},\vec{p}}$ for $\vec{k}$ close to $0$.  We write $\nabla$ for the gradient with respect to $\vec{k}$ and $\partial_{i}$ for partial differentiation with respect to the $i^{\text{th}}$ coordinate of $\vec{k}$.  \emph{No derivatives with respect to $\vec{p}$ appear below}.

The key observation is that the spectral gap for $\wt{L}_{\vec{0}, \vec{p}}$ is preserved in the spectrum of $\wt{L}_{\vec{k} , \vec{p}}$ for $\vec{k}$ sufficiently small.   
 
\begin{lem}\label{lem:analytic}
Given $\epsilon \in (0,\delta)$, with $\delta$ as in Lemma \ref{lem:L0},  there exists  $r$ such that if $|\vec{k}| < r $ then, for each $\vec{p} \in \bb{T}^{d}_{N}$,
\begin{enumerate}
\item $\wt{L}_{\vec{k} , \vec{p}}$ has a single non-degenerate eigenvalue $E_{\vec {p}}(\vec k)$ with $0 \le \Re E_{\vec{p}}(\vec{k}) < \delta -\epsilon$, 
\item The rest of the spectrum of $ \wt{L}_{\vec{k} , \vec{p}}$ is contained in the half plane $   \{ z : \Re z > \delta - \epsilon \}$. 
\end{enumerate}
Furthermore, $E_{\vec {p}}(\vec{k})$ is $C^{2}$ in a neighborhood of $0$,
\begin{equation}\label{eq:explicit0}
E_{\vec {p}}(\vec 0) = 0, \quad \nabla E_{\vec {p}}(\vec 0) = 0, 
\end{equation}
and
\begin{multline}\label{eq:explicit}
\partial_{i} \partial_{j} E_{\vec {p}}(\vec 0) \ = \ 2 \Re \ip{ \partial_{i} \wt{K}_{\vec{0} , \vec{p}} \delta_{0} \otimes 1, [\wt{L}_{\vec{0} , \vec{p}}]^{-1}  \partial_{j} \wt{K}_{\vec{0} , \vec{p}} \delta_{0} \otimes 1} \\
= \  2 \Re \sum_{x,y \in \Z^{d}} x_{i} y_{j} \overline{h(x)} h(y) \ip{ \delta_{[x]_{N}} \otimes 1 ,[\Gamma_{\vec{p}}(0)]^{-1} \delta_{[y]_{N}}\otimes 1} ,
\end{multline}
where $[x]_{N}$ denotes the point in $\Lambda$  equivalent to $x$ modulo $N$ and
\begin{equation}
\Gamma_{\vec{p}}(0) = P_{0}\wt{V}  ( P_{0}^{\perp} \wt{L}_{\vec{0} , \vec{p}}P_{0}^{\perp}  )^{-1} \wt{V} P_{0}.
\end{equation}
  In particular, $\partial_{i} \partial_{j} E_{\vec {p}}(\vec 0)$ is positive definite.  
\end{lem}
\begin{proof} These are essentially standard facts from analytic perturbation theory.  The key point is that
\begin{equation}
\norm{\wt{L}_{\vec{k} , \vec{p}} - \wt{L}_{\vec{0} , \vec{p}}} \ \le \ c |\vec{k}|. \label{difference_in_k}
\end{equation}
If the generators $\wt{L}_{\vec{k},\vec{p}}$ were self-adjoint or normal it would now follow that the spectrum moves by no more than a distance $c |\vec{k}|$ for $\vec{k}$ small.  However, $\wt{L}_{\vec{k},\vec{p}}$ need not be normal so we must argue more carefully.

Due to the spectral gap $\delta$ between $0$ and the rest of the spectrum of $\wt{L}_{\vec{0},\vec{p}}$, we can fit  a contour $\cu{C}$ around the origin in the resolvent set.  Then \eqref{difference_in_k} shows that the spectrum cannot cross $\cu{C}$ for small $\vec{k}$.  A convenient choice for $\cu{C}$ is the rectangle
$$ \cu{C}= \left ( \delta -\epsilon + \im [-R,R] \right ) \cup \left ( [-R, \delta -\epsilon] + \im R \right ) \cup \left ( -R + \im [-R,R] \right ) \cup\left ( [-R, \delta -\epsilon] - \im R \right ),$$
with $R$ fixed independent of $\epsilon$, but sufficiently large.  By Lemma \ref{lem:L0}, 
\begin{equation} \label{boundonC} \sup_{\substack{z \in \cu{C} \\
\vec{p} \in \bb{T}^{d}_{N}}}\norm{ (\wt{L}_{\vec{0},\vec{p}} -z)^{-1}} < \infty.
\end{equation} 

Expanding the resolvent of $\wt{L}_{\vec{k},\vec{p}}$ in a Neumann series,
\begin{equation}
 ( \wt{L}_{\vec{k},\vec{p}} - z )^{-1} = \sum_{n=0}^{\infty } ( \wt{L}_{\vec{0},\vec{p}} - z )^{-1} \left [  (\wt{L}_{\vec{0} , \vec{p}} - \wt{L}_{\vec{k} , \vec{p}})  ( \wt{L}_{\vec{0},\vec{p}} - z )^{-1}\right ]^{n},
 \end{equation}
and using  \eqref{difference_in_k} and \eqref{boundonC}, we see that there is $r > 0$ such that if $|\vec{k}| < r$, then $\cu{C}$ is in the resolvent set of $\wt{L}_{\vec{k},\vec{p}}$.   However, the spectrum is a subset of the numerical range and the numerical range of $\wt{L}_{\vec{k},\vec{p}}$ is contained in the set
\begin{equation}
\set{ x+ \im y \ : \ x \ge 0 \ \& \ | y| \le C + \gamma x},
\end{equation}
with $C = 2 \| \wh{h}\|_{\infty} + 2 \|\wt{V}\|$. We conclude that 
\begin{equation}
\sigma(\wt{L}_{\vec{k},\vec{p}}) = \Sigma_{0} \cup \Sigma_{1}
\end{equation}
with $\Sigma_{0}$ inside $\cu{C}$ and $\Sigma_{1} \subset \{ z \ : \ \Re z > \delta -\epsilon\}$. 

It remains to show that $\Sigma_{0}$ consists of a non-degenerate eigenvalue and to derive \eqref{eq:explicit0} and \eqref{eq:explicit}. For this purpose, consider the (non-Hermitian)  Riesz projection 
\begin{equation}\label{eq:Qk}
Q_{\vec{k},\vec{p}} \ = \ \frac{1}{2 \pi \im} \int_{\cu{C}} \frac{1}{z - \wt{L}_{\vec{k} , \vec{p}}} \di z.
\end{equation}
The rank of $Q_{\vec{k},\vec{p}}$ is constant so long as $\cu{C}$ remains in the resolvent set.  Thus, $Q_{\vec{k} , \vec{p}}$ is rank one  for $|\vec{k}| <r$ and $\Sigma_{0} = \{ E_{\vec{p}}(\vec{k})\}$ with associated normalized eigenvector $\Phi_{\vec{k},{\vec {p}}}$ in the one-dimensional range of $Q_{\vec{k},\vec{p}}$.  Then, $E_{\vec {p}}(\vec 0)=0$ and $\Phi_{\vec 0 , \vec{p}}= \delta_{0}\otimes 1$. By the Feynman-Hellman formula, 
\begin{equation}
\partial_{i} E_{\vec {p}}(\vec{k}) \ = \ \ip{\Phi_{\vec{k},\vec{p} }, \partial_{i} \wt{L}_{\vec{k} , \vec{p}} \Phi_{\vec{k},\vec{p}}},
\end{equation}
from which it follows that $\grad E_{\vec {p}}(\vec 0) =0$ since $\grad \wt{L}_{\vec{k} , \vec{p}}= \im \grad \wt{K}_{\vec{k} , \vec{p}}$ is off-diagonal in the position basis on $\cu{H}_{0}$.  Similarly, 
\begin{multline}
\partial_{i} \partial_{j} E_{\vec {p}}(\vec{k}) \ = \ \ip{ \Phi_{\vec{k}, \vec{p}}, \partial_{i}\partial_{j} \wt{L}_{\vec{k} , \vec{p}} \Phi_{\vec{k}, \vec{p}}}
 +  \ip{\Phi_{\vec{k}, \vec{p}},  Q_{\vec{k}} \partial_{i} \wt{L}_{\vec{k} ,
\vec{p}}
( E_{\vec{p}}(\vec{k}) - \wt{L}_{\vec{k} , \vec{p}})^{-1} (1-Q_{\vec{k},
\vec{p}}) \partial_{j} \wt{L}_{\vec{k} , \vec{p}} \Phi_{\vec{k}, \vec{p}}}
\\
+  \ip{\Phi_{\vec{k}, \vec{p}},  Q_{\vec{k}} \partial_{j} \wt{L}_{\vec{k} ,
\vec{p}}
( E_{\vec{p}}(\vec{k}) - \wt{L}_{\vec{k} , \vec{p}})^{-1} (1-Q_{\vec{k},
\vec{p}}) \partial_{i} \wt{L}_{\vec{k} , \vec{p}} \Phi_{\vec{k}, \vec{p}}}
\end{multline}
The first term on the r.h.s.\ vanishes at $\vec{k}=0$ and the remaining two terms give \eqref{eq:explicit}. Because the form on the r.h.s of \eqref{eq:explicit} is positive definite,   the non-degeneracy condition on the hopping (Assumption 3) gives that $\partial_{i} \partial_{j} E_{\vec {p}}(\vec{0})$ is positive definite.
\end{proof}

It follows from Lemma \ref{lem:analytic} and the sectoriality \eqref{eq:sectoriality} of $B$ that the semigroup $e^{-t \wt{L}_{\vec{k},\vec{p}}}$ satisfies exponential bounds (see \cite[Lemma 4]{KS}):
\begin{lem}\label{lem:Lkdynamics} Given $\epsilon >0$ there is $C_{\epsilon } <\infty$ such that if $\vec{k}$ is sufficiently small, then 
\begin{equation}
\norm{\e^{-t \wt{L}_{\vec{k} , \vec{p}}}(1- Q_{\vec{k} , \vec{p}})} \le C_{\epsilon}\e^{-t (\delta  - \epsilon )}
\end{equation}
for all $\vec{p}$, where $Q_{\vec{k},\vec{p}}$ is the rank one Riesz projection \eqref{eq:Qk} onto the non-degenerate eigenvector of $\wt{L}_{\vec{k},\vec{p}}$ with eigenvalue near $0$.
\end{lem}

\subsection{Proof of Theorem \ref{thm:main}}
As in \cite{KS}, it suffices to prove the theorem for  $\rho_{0}$ satisfying
\begin{equation}\label{eq:summable}
\sum_{x,y} |\rho_{0}(x,y)| < \infty,
\end{equation}
since any initial density matrix can be approximated in trace norm arbitrarily well using such $\rho_{0}$. 
Assuming \eqref{eq:summable}, note that
\begin{equation}
\wt{\rho}_{0;\vec{k}, \vec{p}}(x) \ = \ \sum_{\eta,y \in \Z^d} \rho_{0}(x-N\eta-y ,-y)\e^{\im \vec{p}\cdot (x-N \eta )-\im \vec{k}\cdot y} \label{fourier_transform_rho}
\end{equation}
is uniformly bounded in $\ell^{2}(\Lambda)$ as $\vec{p}$ varies through the torus:
\begin{equation}
\left [ \sum_{x}\abs{\wt \rho_{0;\vec{k}, \vec{p}}(x)}^{2} \right ]^{\half} \ \le \ \sum_{x}\abs{\wt \rho_{0;\vec{k}, \vec{p}}(x)} \ \le \  \sum_{x,y} \abs{\rho_{0}(x,y)}  < \infty.
\end{equation}
By \eqref{E(Rho3)FT}, we have
\begin{align} \label{eq:projectionsdecompose}
   \sum_{x} \e^{-\im \frac{1}{\sqrt{\tau}} \vec{k} \cdot x} \Ev{\rho_{\tau t}(x,x)} \  & = \  \frac{N^{d}}{(2\pi)^{d}}\int_{ \bb{T}_{ N}^d}\di \vec{p}\ip{ \delta_{0} \otimes 1, \e^{-\tau t \wt{L}_{ \vec{k}/\sqrt{\tau}, \vec{p}}} \wt{\rho}_{0;\frac{1}{\sqrt{\tau}}\vec{k}, \vec {p}}\otimes 1}  \\
&  = \ \frac{N^{d}}{(2\pi)^{d}} \int_{\bb{T}_{ N}^d}\di\vec{p} \ \e^{-\tau t E_{\vec {p}}(\vec{k}/\sqrt{\tau})} \ip{\delta_{0}\otimes 1, Q_{\frac{1}{\sqrt{\tau}}\vec{k},\vec{p}} \wt{\rho}_{0;\frac{1}{\sqrt{\tau}} \vec{k}, \vec {p}} \otimes 1}  \label {1st_term}\\ 
&  \quad  + \ \frac{N^{d}}{(2\pi)^{d}} \int_{\bb{T}_{ N}^d}\di\vec{p} \ \ip{ \delta_{0} \otimes 1, \e^{-\tau t \wt{L}_{ \vec{k}/\sqrt{\tau}, \vec{p}}}  (1 - Q_{\frac{1}{\sqrt{\tau}}\vec{k},\vec{p}}) \wt{\rho}_{0;\frac{1}{\sqrt{\tau}} \vec{k}, \vec {p}}\otimes 1}. \label{2nd_term}
\end{align}By Lemma \ref{lem:Lkdynamics},  
the integrand  in \eqref{2nd_term} is exponentially small in the large $\tau$ limit,
\begin{multline}\label{eq:decay}
\abs{\ip{ \delta_{0} \otimes 1,  (1 - Q_{\frac{1}{\sqrt{\tau}}\vec{k} , \vec{p}}) \e^{-\tau t \wt{L}_{ \vec{k}/\sqrt{\tau}, \vec{p}}} \wt{\rho}_{0;\frac{1}{\sqrt{\tau}} \vec{k}, \vec {p}}\otimes 1}}
\\ \le  \  \norm{(1 - Q_{\frac{1}{\sqrt{\tau}}\vec{k} , \vec{p}}) \e^{-\tau t \wt{L}_{ \vec{k}/\sqrt{\tau}, \vec{p}}}}   \norm{ \wt{\rho}_{0;\frac{1}{\sqrt{\tau}} \vec{k}, \vec {p}}\otimes 1} \
\le \ C_{\epsilon}  \e^{-\tau t (\delta-\epsilon  )} \ \ra \ 0.
\end{multline}

Regarding \eqref{1st_term}, we have by Taylor's formula, 
\begin{equation}
E_{\vec {p}}(\vec{k}/\sqrt{\tau}) \ = \ \frac{1}{2\tau} \sum_{i,j} \partial_{i} \partial_{j} E_{\vec {p}}(\vec 0)  \vec{k}_{i} \vec{k}_{j} \ + \ o\left ( \frac{1}{\tau} \right ),
\end{equation}
since $E_{\vec {p}}(\vec 0) = \nabla E_{\vec {p}}(\vec 0)=0$. Thus
\begin{equation}\label{eq:taylor}
\e^{-\tau t E_{\vec {p}}(\vec{k}/\sqrt{\tau})} \ = \ \e^{-t  \frac{1}{2} \sum_{i,j} \partial_{i} \partial_{j} E_{\vec {p}}(\vec 0)  \vec{k}_{i} \vec{k}_{j}} + o(1),
\end{equation}
and 
\begin{multline}\label{eq:final}
\sum_{x} \e^{-\im \frac{1}{\sqrt{\tau}} \vec{k} \cdot x} \Ev{\rho_{\tau t}(x,x)}  \ = \ \frac{N^{d}}{(2\pi)^{d}} \int_{\bb{T}_{ N}^d}\di\vec{p} \ \e^{-t \half \sum_{i,j} \partial_{i} \partial_{j} E_{\vec {p}}(\vec 0)  \vec{k}_{i} \vec{k}_{j}} \ip{\delta_{0}\otimes 1,  \wt{\rho}_{0;\frac{1}{\sqrt{\tau}} \vec{k}, \vec{p}}\otimes 1} + o(1)  \\  \xrightarrow[]{\tau \ra \infty} \ \frac{N^{d}}{(2\pi)^{d}} \int_{\bb{T}_{ N}^d}\di\vec{p}\  \e^{-t  \half \sum_{i,j} \partial_{i} \partial_{j} E_{\vec {p}}(\vec 0)  \vec{k}_{i} \vec{k}_{j}} \wt{\rho}_{0;\vec 0, \vec p}(0)
\end{multline}
since $Q_{\vec{k}, \vec{p}}^{\dagger} \delta_{0} \otimes 1 \ra \delta_{0} \otimes 1$ as $\vec{k} \ra 0$ and $\wt{\rho}_{0;\vec{k}, \vec{p}}(0)$ is continuous as a function of $\vec{k}$.   Letting $D_{i,j}(\vec{p})= \frac{1}{2} \partial_{i} \partial_{j} E_{\vec{p}}(\vec 0)$ and  $m(\vec{p}) = \frac{N^{d}}{(2\pi)^{d}} \wt{\rho}_{0;\vec 0, \vec p}(0)$ gives \eqref{eq:main} and completes the proof.\qed

	 % \bibliography{Bib}
\providecommand{\bysame}{\leavevmode\hbox to3em{\hrulefill}\thinspace}
\providecommand{\MR}{\relax\ifhmode\unskip\space\fi MR }
\providecommand{\MRhref}[2]{%
  \href{http://www.ams.org/mathscinet-getitem?mr=#1}{#2}
}
\providecommand{\href}[2]{#2}

\end{document}